\documentclass[letterpaper, 10 pt, conference]{ieeeconf}

\newtheorem{thm}{Theorem}[section]

\newtheorem{prop}[thm]{Proposition}
\newtheorem{problem}{Problem}
\newtheorem{defn}[thm]{Definition}
\newtheorem{rem}[thm]{Remark}

\IEEEoverridecommandlockouts
\usepackage{graphicx} 
\graphicspath{{figures/}}

\usepackage{epsfig} 
\usepackage{amsmath}
\usepackage{amssymb}
\usepackage{epstopdf}
\usepackage{cite}
\usepackage{algorithm}
\usepackage{algorithmic}
\usepackage{multirow}
\usepackage{rotating}
\usepackage{subfigure}
\usepackage{color}
\usepackage{mysymbol}
\usepackage{url}
\usepackage{ltlfonts}	
	
\begin{document}

\title{\LARGE \bf Simultaneous Intermittent Communication Control and Path Optimization in Networks of Mobile Robots}

\author{Yiannis Kantaros and Michael M. Zavlanos
\thanks{This work is supported in part by the NSF awards CNS $\#$1261828 and CNS $\#$1302284.}
\thanks{Yiannis Kantaros and Michael M. Zavlanos are with the Department of Mechanical Engineering and Materials Science, Duke University, Durham, NC 27708, USA. $\left\{\text{yiannis.kantaros,michael.zavlanos}\right\}$@duke.edu}
}

\maketitle
\thispagestyle{empty}
\pagestyle{empty}
\begin{abstract}
In this paper, we propose an intermittent communication framework for mobile robot networks. Specifically, we consider robots that move along the edges of a connected mobility graph and communicate only when they meet at the nodes of that graph giving rise to a dynamic communication network. Our proposed distributed controllers ensure intermittent connectivity of the network and path optimization, simultaneously. We show that the intermittent connectivity requirement can be encapsulated by a global Linear Temporal Logic (LTL) formula. Then we approximately decompose it into local LTL expressions which are then assigned to the robots. To avoid conflicting robot behaviors that can occur due to this approximate decomposition, we develop a distributed conflict resolution scheme that generates non-conflicting discrete motion plans for every robot, based on the assigned local LTL expressions, whose composition satisfies the global LTL formula. By appropriately introducing delays in the execution of the generated motion plans we also show that the proposed controllers can be executed asynchronously.


\end{abstract}

\section{Introduction}

Communication among robots has been typically modeled using proximity graphs and the communication problem is often treated as preservation of graph connectivity \cite{zavlanos2007potential,sabattini2013decentralized,Zavlanos_IEEETRO08,Zavlanos_IEEE11}. 
Common in the above works is that point-to-point or end-to-end network connectivity is required to be preserved for all time. However, this requirement is often very conservative, since limited resources may hinder robots from accomplishing their assigned goals. Motivated by this fact, in this paper we propose a \textit{distributed intermittent communication protocol} for mobile networks. In particular, we consider that robots move along the edges of a connected mobility graph and communicate only when they meet at the nodes of this graph giving rise to a dynamic communication network. We design distributed controllers that ensure intermittent communication of the network while minimizing at the same time the distance traveled between meeting points. We show that intermittent communication can be captured by a global Linear Temporal Logic (LTL) formula that forces robots to meet infinitely often at the rendezvous points. Given such an LTL expression, existing model checking techniques \cite{baier2008principles,clarke1999model} can be employed in order to implement correct by construction controllers for all robots. 

LTL-based control synthesis and task specification for mobile robots build upon either a bottom-up approach when independent LTL expressions are assigned to robots \cite{loizou2004automatic,kress2009temporal,guo2014cooperative} or top-down approaches when a global LTL describing a collaborative task is assigned to a team of robots \cite{chen2011synthesis,kloetzer2008distributed}, as in our work. Bottom-up approaches generate a discrete high-level motion plan for all robots based on a synchronous product automaton among all agents and, therefore, they are resource demanding and scale poorly with the number of robots. To mitigate these issues, we propose a novel technique that approximately decomposes the global LTL formula into local ones and assigns them to robots. Since the approximate decomposition of the global LTL formula can result in conflicting robot behaviors we develop a distributed conflict resolution scheme that generates discrete motion plans for every robot based on the assigned local LTL expressions. By appropriately introducing delays in the execution of the generated motion plans we show the proposed controllers can also be executed in an asynchronous fashion. In contrast, most relevant literature assumes that robot control is performed in a synchronous way \cite{chen2011synthesis,kloetzer2008distributed}. Asynchronous robot mobility is considered in \cite{ulusoy2013optimality} by introducing `traveling states' in the discretized abstraction of the environment decreasing in this way the scalability of the proposed algorithm.

     
The most relevant works to the one proposed here are presented in \cite{kantaros15asilomar,kantaros16acc,zavlanos2010synchronous,hollinger2010multi}. A centralized intermittent communication control scheme is presented in \cite{kantaros15asilomar} that ensures communication among robots infinitely often but it does not scale well with the number of robots.  In \cite{kantaros16acc} a distributed intermittent communication control scheme is proposed that requires synchronization among robots, unlike the approach developed here. Moreover, \cite{kantaros16acc} considers \textit{a priori} determined communication points while in this paper, rendezvous points are selected optimally, to minimize the distance traveled by the robots. \cite{zavlanos2010synchronous} proposes a distributed synchronization scheme that allows robots that move along the edges of a bipartite mobility graph to meet periodically at the vertices of this graph. Instead, here we make no assumptions on the graph structure on which robots reside or on the communication pattern to be achieved. On the other hand, \cite{hollinger2010multi} proposes a receding horizon framework for periodic connectivity that ensures recovery of end-to-end connectivity within a given time horizon. As the number of robots or the size of the time horizon grows, this approach can become computationally expensive. To the contrary, our proposed method scales well to large numbers of robots and can handle situations where the whole network can not be connected at once, by ensuring connectivity over time, infinitely often.

\section{Problem Formulation}\label{sec:prob}

Consider a team of $N$ robots that move in a workspace $\ccalW\subset\mathbb{R}^n$ according to $\dot{\bbx}_i(t)=\bbu_i(t)$, where $\bbx_{i}(t)\in\mathbb{R}^n$ is the position of robot $i$, $i\in\{1,2,\dots,N\}$, at time $t$ and $\bbu_i(t)\in\mathbb{R}^n$ is a control input. Also, consider $L$ locations in $\ccalW$ denoted by $\ell_j$, $j\in\{1,2\dots,L\}$ located at positions $\bbmu_{j}\in\ccalW$ and paths $\gamma_{ij}:~[0,1]\rightarrow \mathbb{R}^n$ that connect two locations $\ell_i$ and $\ell_j$ such that $\gamma_{ij}(0)=\bbmu_{i}$ and $\gamma_{ij}(1)=\bbmu_{j}$. The union of locations $\bbmu_{j}$ and paths $\gamma_{ij}$ gives rise to an undirected graph $\ccalG=\{\mathcal{V},\mathcal{E}\}$, where the set of nodes $\mathcal{V}=\left\{1,2,\dots,L\right\}$ is indexed by the set of locations $\ell_j$ and the set of edges $\mathcal{E}\subseteq\mathcal{V}\times\mathcal{V}$ is determined by the paths $\gamma_{ij}$ such that an edge $(i,j)\in\ccalE$ exists if and only if a path $\gamma_{ij}$ exists. Hereafter, we call $\ccalG$ a mobility graph and we assume that the robots move along its edges to possibly accomplish a high-level task; see Figure \ref{mobility}. For example, robots may travel along their assigned paths to monitor different parts of a region and then coordinate to transmit their measurements to a user. The presence of a network allows the robots to transmit data in a multi-hop fashion, so that they do not have to leave their assigned region. Applications of this framework involve distributed coverage, state estimation, or surveillance. In what follows, we assume that the mobility graph $\ccalG$ is connected. 

\subsection{Intermittent Communication}\label{sec:intermittent}

We assume that the robotic team is divided into $M$ subgroups. The indices $i$ of robots that belong to the $m$-th subgroup are collected in a set denoted by $\mathcal{T}_m$, $m\in\{1,2,\dots,M\}$ while every robot can belong to more than one subgroups. Robots in a subgroup $\mathcal{T}_m$ can communicate only when all of them are present simultaneously at a common location $\ell_j$. The locations $\ell_j$ where communication can take place for the robotic team $\ccalT_m$ are collected in a set $\ccalC_m$.  This way, a dynamic robot communication graph $\ccalG_c=\{\mathcal{V}_c,\mathcal{E}_c\}$ is constructed where the set of nodes $\ccalV_c$ is indexed by robots, i.e., $\ccalV_c=\left\{1,2,\dots,N\right\}$ and $\ccalE_c\subseteq\ccalV_c\times\ccalV_c$ is the set of communication links that emerge when, e.g., all robots in a team $\ccalT_m$ meet at a common rendezvous point $\ell_j\in\ccalC_m$ simultaneously. Given the robot teams $\ccalT_m$, a graph $\ccalG_{\ccalT}=\{\ccalV_{\ccalT},\ccalE_{\ccalT}\}$ is constructed whose set of nodes $\ccalV_{\ccalT}=\{1,2,\dots,M\}$ is indexed by the teams $\ccalT_m$ and set of edges $\ccalE_{\ccalT}$ consists of links between nodes $m$ and $n$ if $\ccalT_m\cap\ccalT_n\neq\varnothing$; see Figure 1(b). We assume that the graph $\ccalG_{\ccalT}$ is connected in order to ensure dissemination of information in the network. Furthermore, we can define the set of neighbors of node $m\in\ccalV_{\ccalT}$ by $\ccalN_{\ccalT_m}=\left\{n\in\ccalV_{\ccalT}|(n,m)\in\ccalE_{\ccalT}\right\}$. Also, for robot $i$ we define the set that collects the indices of all other robots that belong to common teams with robot $i$ as $\ccalN_i=\{j|j\in\ccalT_m, \forall m\in\ccalS_i\}\setminus\{i\}$, and the set that collects the indices of teams that robot $i$ belongs to as $\ccalS_i=\{m|i\in\ccalT_m,~m\in\{1,2,\dots,M\}\}$.

At a rendezvous point $\ell_j\in\ccalC_m$ communication takes place when all robots in $\ccalT_m$ are present there, simultaneously. Then the communication graph $\ccalG_c$ is defined to be \textit{connected over time} if all robots in $\ccalT_m$ meet infinitely often at a region $\ell_j\in\ccalC_m$, for all $m\in\ccalV_{\ccalT}$. Such a requirement can be captured by the following global LTL expression:

\begin{equation}\label{eq:globalLTL}
\phi=\wedge_{\forall m\in\{1,2,\dots,M\}}\left(\square\Diamond\left(\vee_{\ell_j\in\ccalC_m}(\wedge_{\forall i\in\ccalT_m}\pi_{i}^{\ell_j})\right)\right),
\end{equation}
where $\pi_{i}^{\ell_j}$ is an atomic proposition that is true when robot $i$ is sufficiently close to $\bbmu_{j}$. 
In \eqref{eq:globalLTL}, $\wedge$ and $\vee$ are the conjunction and disjunction operator, respectively, while $\square$ and $\Diamond$ stand for the temporal operators `always' and `eventually', respectively. For more details on LTL, we refer the reader to \cite{baier2008principles,clarke1999model}. In words, the LTL expression in \eqref{eq:globalLTL} requires all robots in a team $\ccalT_m$ to meet infinitely often at at least one communication point $\ell_j\in\ccalC_m$, for all $m\in\ccalV_{\ccalT}$.

\begin{figure}[tb]
    \centering
    \subfigure[Mobility Graph $\ccalG$]{
    \label{mobility}
     \includegraphics[width=0.47\linewidth]{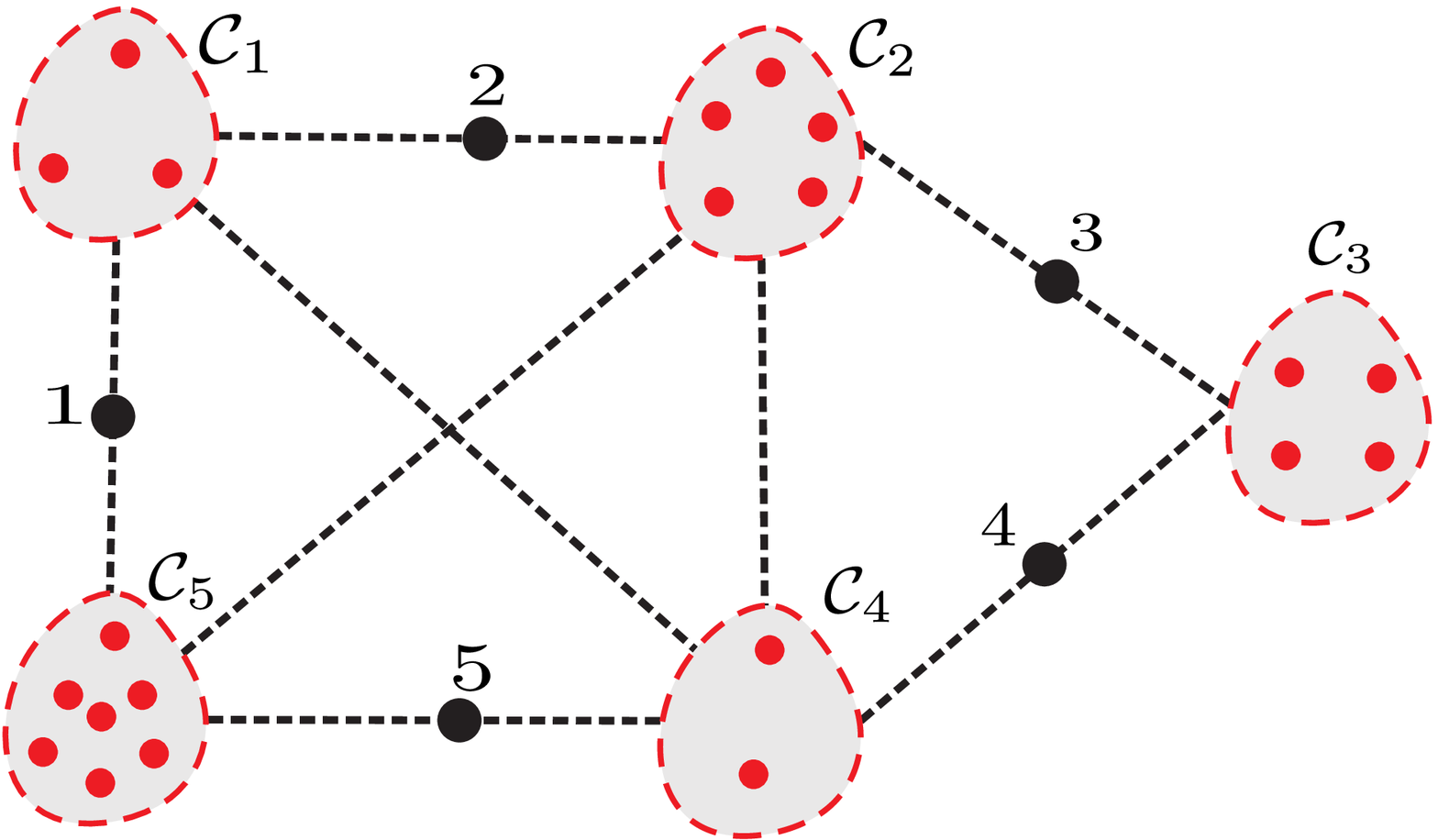}}
      \subfigure[Graph $\ccalG_{\ccalT}$]{
    \label{graphgt}
     \includegraphics[width=0.47\linewidth]{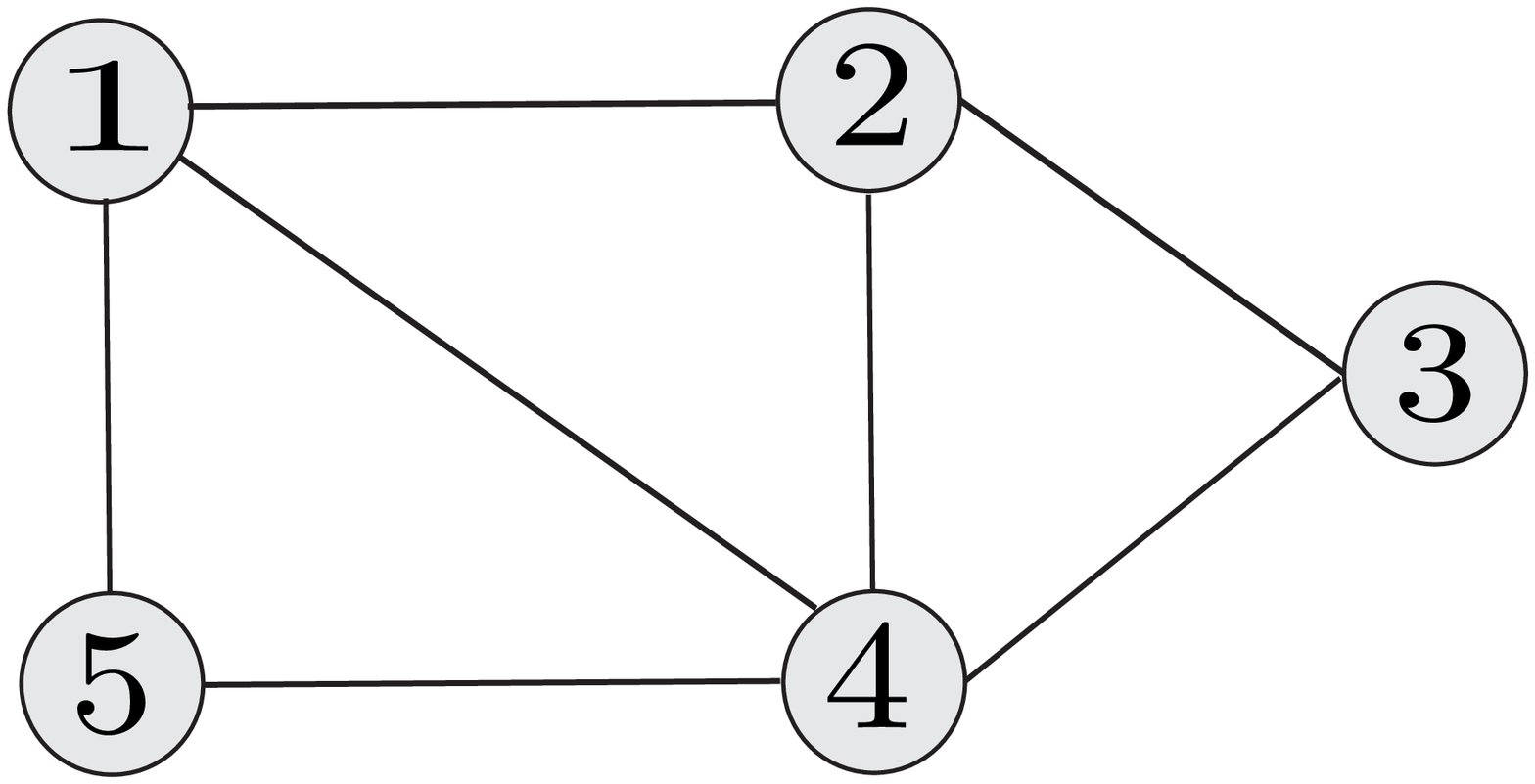}}
  \caption{A graphical illustration of the problem formulation. Figure \ref{mobility} depicts the mobility graph $\ccalG$ for a network of $N=5$ robots (black dots) divided into $M=5$ teams. The robot teams are selected to be: $\ccalT_1=\{1,2\}$, $\ccalT_2=\{2,3\}$, $\ccalT_3=\{3,4\}$, $\ccalT_4=\{2,4,5\}$, and $\ccalT_5=\{1,5\}$. Red dots represent communication points and the black dashed lines between two sets $\ccalC_m$ and $\ccalC_n$ imply that there is at least one path $\gamma_{ij}$ that connects communication points $\ell_i\in\ccalC_m$ and $\ell_j\in\ccalC_n$. Figure \ref{graphgt} depicts the associated graph $\ccalG_{\ccalT}$.}\label{geod}
\end{figure} 

\subsection{Discretized Abstraction of the Workspace}
We can model the environment in which robot $i$ resides by a weighted transition system (wTS) denoted by $\text{wTS}_{i}$ that is defined as follows
\begin{defn}[weighted Transition System]\label{defn:wTS}
A weighted \textit{Transition System} $\text{wTS}_{i}$ is a tuple $\left(\ccalQ_{i}, q_{i}^0,\ccalA_{i},\rightarrow_{i}, w_i, \mathcal{AP},L_{i}\right)$ where: (a) $\ccalQ_{i}=\{q_{i}^{\ell_j}\}_{\forall\ell_j\in\ccalC_m,\forall m\in\ccalS_i}$ is the set of states, where a state $q_{i}^{\ell_j}$ indicates that robot $i$ is at $\ell_j$; (b) $q_{i}^0\in\ccalQ_{i}$ is the initial state of robot $i$; (c) $\ccalA_{i}$ is a set of actions. The available actions at state $q_{i}^{\ell_j}$ are ``wait'' and ``go to state $q_{i}^{\ell_k}$'' for every $k$ such that $(j,k)\in\ccalE$;\footnote{Throughout the rest of the paper, for the sake of simplicity we assume that the mobility graph $\ccalG$ is constructed so that there is a path from any location $\ell_j$ to any $\ell_k$, where $q_i^{\ell_j}\in\ccalQ_i$ and $q_i^{\ell_k}\in\ccalQ_i$, that does not pass though a location $\ell_e\in\ccalC_m$ for some $m\notin\ccalS_i$, for all robots $i$. This assumption can be relaxed if we incorporate all communication points $\ell_j$ in the set of states $\ccalQ_i$. }(d) $\rightarrow_{i}\subseteq\ccalQ_{i}\times\ccalA_{i}\times\ccalQ_{i}$ is the transition relation; (e) $w_{i}:\ccalQ_{i}\times\ccalQ_{i}\rightarrow \mathbb{R}_+$ is a cost function that assigns weights/cost to each possible transition in wTS. These costs are associated with the distance between two states $q_i^{\ell_j}$ and $q_i^{\ell_k}$; (f) $\mathcal{AP}$ is the set of atomic propositions; and (h) $L_{i}:\ccalQ_{i}\rightarrow 2^{\mathcal{AP}}$ is an observation/output relation giving the set of atomic propositions that are satisfied in a state. 
\end{defn} 

In what follows we give definitions related to $\text{wTS}_{i}$, that we will use throughout the rest of the paper.

\begin{defn}[Infinite Path]\label{def:infpath}
An \textit{infinite path} $\tau_{i}$ of $\text{wTS}_{i}$ is an infinite sequence of states, $\tau_{i}=\tau_{i}(1)\tau_{i}(2)\tau_{ij}(3)\dots$ such that $\tau_{i}(1)=q_{i}^0$, $\tau_{i}(n)\in\ccalQ_{i}$, and $(\tau_{i}(n),a_{i},\tau_{i}(n+1))\in\rightarrow_{i}$, for some $a_{i}\in\ccalA_{i}$, $\forall n$.\footnote{A \textit{finite path} of $\text{wTS}_{i}$ can be defined accordingly. The only difference with the infinite path is that a finite path is defined by an finite sequence of states of $\text{wTS}_{i}$.}
\end{defn}

\begin{defn}[Composition]\label{def:compos}
Composition of $M$ infinite paths $\tau_m=\tau_m(1)\tau_m(2)\tau_m(3)\dots$, where $m\in\left\{1,\dots,M\right\}$, denoted by $\tau=\otimes_{\forall m}\tau_m$ is an infinite sequence of states defined as $\tau=\tau(1)\tau(2)\dots=\left[\tau(n)\right]_{n=1}^{\infty}$, where $\tau(n)=(\tau_1(n),\tau_2(n),\dots,\tau_M(n))$. 
\end{defn}

\begin{defn}[Cost]\label{def:cost}
The cost of an infinite path $\tau_i$ of a $\text{wTS}_i$ is $J_i(\tau_i)=\sum\nolimits_{n=1}^{\infty}w_i(\tau_i(n),\tau_i(n+1))$. Similarly, the cost incurred by the composition of $N$ infinite paths $\tau_i$ denoted by $\tau=\otimes_{\forall i\in\{1,2,\dots,N\}}\tau_i$ is given by
\end{defn}
\vspace{-1mm}
\begin{align}\label{eq:cost1}
J(\tau)&=
\sum\nolimits_{n=0}^{\infty}\sum\nolimits_{i=1}^{N}w_i(\tau_i(n),\tau_i(n+1)),
\end{align} 
%

\begin{defn}[Projection]
For an infinite path $\tau=\tau(1)\tau(2)\tau(3)\dots$, we denote by $\Pi|_{\text{wTS}_{i}}\tau$ its projection onto $\text{wTS}_{i}$, which is obtained by erasing all states in $\tau$ that do not belong to $\ccalQ_{i}$. 
\end{defn}

\begin{defn}[Trace of infinite path]\label{def:tracet}
The trace of an infinite path $\tau_{i}=\tau_{i}(1)\tau_{i}(2)\tau_{i}(3)\dots$ of a wTS $\text{wTS}_{i}$, denoted by $\texttt{trace}(\tau_{i})$, is an infinite word that is determined by the sequence of atomic propositions that are true in the states along $\tau_{i}$, i.e., $\texttt{trace}(\tau_{i})=L_{i}(\tau_{i}(1))L_{i}(\tau_{i}(2))\dots$. 
\end{defn}


\begin{defn}[Motion Plan]\label{def:motionplan}
Given an LTL formula $\phi$, a wTS $\text{wTS}_{i}$ both defined over the set of atomic propositions $\mathcal{AP}$, an infinite path $\tau_{i}$ of $\text{wTS}_{i}$ is called \textit{motion plan} if and only if $\texttt{trace}(\tau_{i})\in\texttt{Words}(\phi)$, where $\texttt{Words}(\phi)=\left\{\sigma\in (2^{\mathcal{AP}})^{\omega}|\sigma\models\phi\right\}$ is defined as the set of words $\sigma\in (2^{\mathcal{AP}})^{\omega}$ that satisfy the LTL $\phi$ and $\models\subseteq (2^{\mathcal{AP}})\times\phi$ is the satisfaction relation. The relation $\texttt{trace}(\tau_{i})\in\texttt{Words}(\phi)$ is equivalently denoted by $\tau_{i}\models\phi$.
\end{defn}

The problem we address in this paper can be stated as:

\begin{problem}\label{pr:pr1}
Given any initial configuration of the robots in the mobility graph $\ccalG$ determine motion plans $\tau_{i}$ for all robots $i$ that satisfy global LTL expression given in \eqref{eq:globalLTL}, i.e., communication graph $\ccalG_c$ is connected over time and minimize the total distance traveled by robots captured by the objective function \eqref{eq:cost1}.
\end{problem}

\section{Intermittent Communication Control}\label{sec:control}

To solve Problem \ref{pr:pr1}, known centralized model checking techniques can be employed, that typically rely on a discretized abstraction of the environment captured by a wTS and the construction of a synchronized product system among all robots of the network. As a result, such approaches are resource demanding and scale poorly with the size of the network. Therefore, a distributed solution is preferred whereby discrete high-level motion plans for every robot can be computed locally across the network. For this purpose, notice first that although the global LTL formula \eqref{eq:globalLTL} is not decomposable with respect to robots, it can be decomposed in local LTL formulas $\phi_{\ccalT_m}$ associated with a robot team $\ccalT_m$, which are coupled with each other by the conjunction operator $\wedge$. Specifically, we can write $\phi=\wedge_{m\in\ccalV_{\ccalT}}\phi_{\ccalT_m}$, 
%
where $\phi_{\ccalT_m}$ is defined as
\begin{equation}\label{eq:localLTL}
\phi_{\ccalT_m}= \square\Diamond\left(\vee_{\ell_j\in\ccalC_m}(\wedge_{\forall i\in\ccalT_m}\pi_{i}^{\ell_j})\right),
\end{equation}
and forces all robots $i\in\ccalT_m$ to meet infinitely often at at least one rendezvous point $\ell_j\in\ccalC_m$.  

Given the decomposition of $\phi$ into local LTL formulas $\phi_{\ccalT_m}$, every robot $i\in\ccalT_m$ needs to develop motion plans $\tau_{i}$ so that the composition of plans $\tau_{i}$, $\forall i\in\ccalT_m$ denoted by $\tau_{\ccalT_m}=\otimes_{i\in\ccalT_m}\tau_{i}$ satisfies the local LTL expression $\phi_{\ccalT_m}$, for all $m\in\ccalS_i$. In this way, we can ensure that the composition of $\tau_{i}$, $\forall i\in\{1,2\dots,N\}$, satisfies the global LTL expression \eqref{eq:globalLTL}, since all local LTL expressions $\phi_{\ccalT_m}$ are satisfied. 

Motion plans $\tau_{\ccalT_m}\models\phi_{\ccalT_m}$, $\forall m\in\ccalV_{\ccalT}$, can be constructed using existing tools from model checking theory \cite{baier2008principles,clarke1999model}. However, notice that constructing plans $\tau_{\ccalT_m}$ and $\tau_{\ccalT_n}$, $\forall n\in\ccalN_{\ccalT_m}$ independently cannot ensure that the robots' behavior in the workspace will satisfy the global LTL formula \eqref{eq:globalLTL}. The reason is that the local LTL formulas $\phi_{\ccalT_m}$ in \eqref{eq:localLTL} are not independent from the local LTL expressions $\phi_{\ccalT_n}$ for which it holds $n\in\ccalN_{\ccalT_m}$, since they are coupled by robots' state in their respective transition systems. For instance, assume that $n\in\ccalN_{\ccalT_m}$ and that robot $i$ belongs to teams $\ccalT_m$ and $\ccalT_n$. Then robot $i$ is responsible for communicating with the other robots that belong to teams $\ccalT_m$ and $\ccalT_n$ at a communication point in $\ccalC_m$ and $\ccalC_n$, respectively. This equivalently implies that the LTL expressions $\phi_{\ccalT_m}$ and $\phi_{\ccalT_n}$ are coupled due to robot $i$ through the atomic propositions $\pi_{i}^{\ell_j}$, $\forall \ell_j\in\ccalC_m$ and $\pi_{i}^{\ell_k}$, $\forall \ell_k\in\ccalC_n$. Consequently, generating plans $\tau_{\ccalT_m}\models\phi_{\ccalT_m}$ that ignore the LTL expressions $\phi_{\ccalT_n}$, $\forall n\in\ccalN_{\ccalT_m}$ may result in conflicting robot behaviors, since the projection of motion plans $\tau_{\ccalT_m}$ and $\tau_{\ccalT_n}$ onto $\text{wTS}_{i}$ may result in two different motion plans $\tau_{i}$ for a robot $i\in\ccalT_m\cap\ccalT_n$, $n\in\ccalN_{\ccalT_m}$. This means that cases where a robot $i$ needs to be in more than one states in $\text{wTS}_{i}$ simultaneously may occur. 

To circumvent these issues, we propose a distributed algorithm in Section \ref{sec:conflict} that implements free-of-conflict discrete motion plans $\tau_{i}$, for all robots $i$, so that the global LTL expression $\phi$ is satisfied. These motion plans will be constructed based on the prefix parts of motion plans $\tau_{\ccalT_m}\models\phi_{\ccalT_m}$, constructed in Section \ref{sec:autmodel}, for all $m\in\ccalS_i$.   

\subsection{Optimal Automata-based Model Checking}\label{sec:autmodel}
 
Given an LTL formula $\phi_{\ccalT_m}$ and the $\text{wTS}_{i}$ of all robots $i\in\ccalT_m$ a motion plan $\tau_{\ccalT_m}\models\phi_{\ccalT_m}$ can be implemented using existing automata-based model checking methods \cite{baier2008principles,clarke1999model}. First the \textit{weighted Product Transition System} (\text{wPTS}) $\text{wPTS}_{\ccalT_m}$ is constructed, which essentially captures all the possible combinations of robots' states in their respective $\text{wTS}_{i}$, $\forall i\in\ccalT_m$ and is defined as follows:

\begin{defn}[Product Transition System]
Given $|\ccalT_m|$ weighted Transition Systems $\text{wTS}_{i_k}=\left(\ccalQ_{i_k}, q_{i_k}^0,\ccalA_{i_k},\rightarrow_{i_k},w_{i_k},\mathcal{AP},L_{i_k}\right)$, where $i_k\in\ccalT_m$, $k=1,2,\dots,|\ccalT_m|$, the \textit{weighted Product Transition System} $\text{wPTS}_{\ccalT_m}=\text{wTS}_{i_1}\otimes\text{wTS}_{i_2}\otimes\dots\otimes\text{wTS}_{i_{|\ccalT_m|}}$ is a tuple $\left(\ccalQ_{\ccalT_m}, q_{\ccalT_m}^0,\ccalA_{\ccalT_m},\longrightarrow_{\ccalT_m},w_{\ccalT_m},\mathcal{AP},L_{\ccalT_m}\right)$ where: (a)$\ccalQ_{\ccalT_m}=\ccalQ_{i_1}\times\ccalQ_{i_2}\times\dots\times\ccalQ_{i_{|\ccalT_m|}}$ is the set of states; (b) $q_{\ccalT_m}^0=(q_{i_1}^0,q_{i_2}^0,\dots,q_{i_{|\ccalT_m|}}^0)\in\ccalQ_{\ccalT_m}$ is the initial state; (c)$\ccalA_{\ccalT_m}=\ccalA_{i_1}\times\ccalA_{i_2}\times\dots\times\ccalA_{i_{|\ccalT_m|}}$ is a set of actions; (d)$\longrightarrow_{\ccalT_m}\subseteq\ccalQ_{\ccalT_m}\times\ccalA_{\ccalT_m}\times\ccalQ_{\ccalT_m}$ is the transition relation defined by the rule\footnote{The notation of this rule is along the lines of the notation used in \cite{baier2008principles}. In particular, it means that if the proposition above the solid line is true, then so does the proposition below the solid line.} $\frac{\bigwedge _{\forall i_k}\left(q_{i_k}\xrightarrow{a_{i_k}}_{i_k}q_{i_k}^{'}\right)}{q_{\ccalT_m}\xrightarrow{a_{\ccalT_m}=\left(a_{i_1},\dots,a_{i_{|\ccalT_m|}}\right)}_{\ccalT_m}q_{\ccalT_m}^{'}}$;\footnote{The state $q_{\ccalT_m}$ stands for the state $\left(q_{i_1},\dots,q_{i_{|\ccalT_m|}}\right)\in\ccalQ_{\ccalT_m}$. The state $q_{\ccalT_m}^{'}$ is defined accordingly.}
(e)$w_{\ccalT_m}(q_{\ccalT_m},q_{\ccalT_m}^{'})=\sum_{k=1}^{\ccalT_m} w_i(\Pi|_{\text{wTS}_{i_k}}q_{\ccalT_m},\Pi|_{\text{wTS}_{i_k}}q_{\ccalT_m}^{'})$; (f)$\mathcal{AP}$ is the set of atomic propositions; and,
(h)$L_{\ccalT_m}=\bigcup_{i_k\in\ccalT_m}L_{i_k}$ is an observation/output relation giving the set of atomic propositions that are satisfied at a state. 
\end{defn}

Next the LTL formula $\phi_{\ccalT_m}$ is translated into a Nondeterministic B$\ddot{\text{u}}$chi Automaton (NBA) over $2^{\mathcal{AP}}$ denoted by $B_{\ccalT_m}$ \cite{vardi1986automata}, which is defined as follows:
\begin{defn}\label{def:buchi}
A \textit{Nondeterministic B$\ddot{\text{u}}$chi Automaton} (NBA) $B_{\ccalT_m}$ over $2^{\mathcal{AP}}$ is defined by the tuple $B_{\ccalT_m}=\left(\ccalQ_{B_{\ccalT_m}}, \ccalQ_{B_{\ccalT_m}}^0,2^{\mathcal{AP}},\rightarrow_{B_{\ccalT_m}},\mathcal{F}_{B_{\ccalT_m}}\right)$ where: (a) $\ccalQ_{B_{\ccalT_m}}$ is the set of states; (b) $\ccalQ_{B_{\ccalT_m}}^0\subseteq\ccalQ_{B_{\ccalT_m}}$ is a set of initial states; (c) $\Sigma=2^{\mathcal{AP}}$ is an alphabet; (d) $\rightarrow_{B_{\ccalT_m}}\subseteq\ccalQ_{B_{\ccalT_m}}\times \Sigma\times\ccalQ_{B_{\ccalT_m}}$ is the transition relation; and (e) $\ccalF_{B_{\ccalT_m}}\subseteq\ccalQ_{B_{\ccalT_m}}$ is a set of accepting/final states. 
\end{defn}

Once the wPTS $\text{wPTS}_{\ccalT_m}$ and the NBA $B_{\ccalT_m}$ that corresponds to the LTL $\phi_{\ccalT_m}$ are constructed, a motion plan $\tau_{\ccalT_m}\models\phi_{\ccalT_m}$ can be found by checking the non-emptiness of the language of the \textit{Product B$\ddot{\text{u}}$chi Automaton} (PBA) $P_{\ccalT_m}=\text{wPTS}_{\ccalT_m}\otimes B_{\ccalT_m}$ \cite{baier2008principles}, which is defined as follows:

\begin{defn}[Product B$\ddot{\text{u}}$chi Automaton]
Given the product transition system $\text{wPTS}_{\ccalT_m}=\left(\ccalQ_{\ccalT_m}, q_{\ccalT_m}^0,\ccalA_{\ccalT_m},\longrightarrow_{\ccalT_m},w_{\ccalT_m},\mathcal{AP},L_{\ccalT_m}\right)$  and the NBA $B_{\ccalT_m}=\left(\ccalQ_{B_{\ccalT_m}}, \ccalQ_{B_{\ccalT_m}}^0,2^{\mathcal{AP}},\rightarrow_{B_{\ccalT_m}},\mathcal{F}_{B_{\ccalT_m}}\right)$, the \textit{Product B$\ddot{\text{u}}$chi Automaton} (PBA) $P_{\ccalT_m}=\text{wPTS}_{\ccalT_m}\otimes B_{\ccalT_m}$ is a tuple $\left(\ccalQ_{P_{\ccalT_m}}, \ccalQ_{{P_{\ccalT_m}}}^0,\longrightarrow_{P_{\ccalT_m}},\ccalF_{P_{\ccalT_m}}\right)$ where: (a) $\ccalQ_{P_{\ccalT_m}}=\ccalQ_{\ccalT_m}\times\ccalQ_{B_{\ccalT_m}}$ is the set of states; (b) $\ccalQ_{P_{\ccalT_m}}^0=q_{\ccalT_m}^0\times\ccalQ_{B_{\ccalT_m}}^0$ is a set of initial states; (c)   $\longrightarrow_{P_{\ccalT_m}}\subseteq\ccalQ_{P_{\ccalT_m}}\times\ccalA_{\ccalT_m}\times 2^{\mathcal{AP}}\times\ccalQ_{P_{\ccalT_m}}$ is the transition relation defined by the rule: $\frac{\left(q_{\ccalT_m}\xrightarrow{a_{\ccalT_m}}q_{\ccalT_m}^{'}\right)\wedge\left( q_{B_{\ccalT_m}}\xrightarrow{L_{\ccalT_m}\left(q_{\ccalT_m}^{'}\right)}q_{B_{\ccalT_m}}^{'}\right)}{q_{P_{\ccalT_m}}=\left(q_{\ccalT_m},q_{B_{\ccalT_m}}\right)\xrightarrow{a_{\ccalT_m}}_{P_{\ccalT_m}}q_{P_{\ccalT_m}}^{'}=\left(q_{\ccalT_m}^{'},q_{B_{\ccalT_m}}^{'}\right)}$; (d) $\ccalF_{P_{\ccalT_m}}=\ccalQ_{\ccalT_m}\times\ccalF_{B_{\ccalT_m}}$ is a set of accepting/final states. 
\end{defn} 

To check the non-emptiness of the language of $P_{\ccalT_m}$ denoted by $\ccalL_{P_{\ccalT_m}}=\texttt{trace}(\text{wPTS}_{\ccalT_m})\cap\mathcal{L}_{B_{\ccalT_m}}$ and to find the motion plan that both satisfies $\phi_{\ccalT_m}$ and at the same time minimizes $J(\tau_{\ccalT_m})$, we can employ existing model checking methods that are based on graph search algorithms; see, e.g., \cite{guo2013motion, guo2013reconfiguration}. Such motion plans can be written in a prefix-suffix structure $\tau_{\ccalT_m}=\tau^{\text{pre}}_{\ccalT_m}[\tau^{\text{suf}}_{\ccalT_m}]^{\omega}$, where the prefix part $\tau^{\text{pre}}_{\ccalT_m}$ is executed only once and the suffix part $\tau^{\text{suf}}_{\ccalT_m}$ is repeated infinitely. In principle, in these approaches the PBA is viewed as a weighted directed graph with weights assigned on each edge that are inherited by the function $w_{\ccalT_m}$. Then finding the shortest path from an initial state to a final state and projecting this path onto $\text{wPTS}_{\ccalT_m}$ results in the prefix part $\tau^{\text{pre}}_{\ccalT_m}$. The suffix part $\tau^{\text{suf}}_{\ccalT_m}$ is constructed similarly by computing the shortest cycle around that final state.  
 
\subsection{Conflict Resolution Coordination}\label{sec:conflict}

As discussed in the beginning of Section \ref{sec:control}, constructing motion plans $\tau_{\ccalT_m}$ for all $m\in\ccalS_i$, for all robots $i$ can result in conflicting robot behaviors. To overcome this issue, we propose a distributed algorithm that resolves any conflicts in the robot behavior introduced by the motion plans $\tau_{\ccalT_m}$ and constructs free-of-conflicts motion plans $\tau_{i}$ for all robots $i$ using the prefix parts $\tau^{\text{pre}}_{\ccalT_m}$ constructed in Section \ref{sec:autmodel}. The general form of these motion plans is  
 
{\small\begin{align}\label{eq:planij}
\tau_{i}=&\tau_i(1)\tau_i(2)\dots=[\tau_i(n)]_{n=1}^{\infty}\nonumber\\=&\bigg[X\dots X\Pi|_{\text{wTS}_{i}}p_{\ccalT_{m_1}}^k X\dots X\Pi|_{\text{wTS}_{i}}p_{\ccalT_{m_j}}^k \nonumber\\&X\dots X\Pi|_{\text{wTS}_{i}}p_{\ccalT_{m_{\left|\ccalS_i\right|}}}^k X\dots X\bigg]_{k=1}^{\infty}=\left[p_{i}^k\right]_{k=1}^{\infty},
\end{align} }
such that $\tau_i(1)=q_i^0$. In \eqref{eq:planij}, $p_{\ccalT_{m_j}}^k$ is a finite path of $\text{wPTS}_{\ccalT_m}$, where $m_j\in\ccalS_i$ and $j\in\{1,\dots,\left|\ccalS_i\right|\}$. Also, $X$ stands for a finite path in which robot $i$ waits at its current state in $\text{wTS}_i$. In \eqref{eq:planij}, the concatenation of the paths $X$ and $p_{\ccalT_{m_j}}^k$, $\forall m_j\in\ccalS_i$ gives rise to the finite path $p_i^k$. Hereafter, the $e$-th finite path in $p_i^k$ is denoted by $p_i^{k,e}$ where $1\leq e\leq \ell$, where $\ell$ stands for the number of finite paths that appear in $p_i^k$. The parameter $\ell$ is \textit{a priori} selected to be $\ell=\max\left\{d_{\ccalT_m}\right\}_{m=1}^M+1$ for all robots, where $d_{\ccalT_m}$ denotes the degree of vertex $m$ in the graph $\ccalG_{\ccalT}$. This particular choice for the parameter $\ell$ ensures the construction of free-of-conflict motion plans, as it will shown in Proposition \ref{prop:ell}. Moreover, we denote by $p_i^k(n)$ the $n$-th state in the finite path $p_i^k$, e.g., $p_i^1(1)=p_i^{1,1}(1)=\tau_i(1)=q_i^0$, by construction of $\tau_i$. The same notation extends to the infinite path $\tau_i$.

In what follows we first describe the construction of the finite paths $p_{\ccalT_{m_j}}^k$ and then we show how these finite paths are ordered in $p_i^k$ giving rise to a free-of-conflict motion plan $\tau_i$. First, for the finite path $p_{\ccalT_{m_j}}^k$ it holds that $p_{\ccalT_{m_j}}^k=\tau^{\text{pre},k}_{\ccalT_{m_j}}$, where $\tau^{\text{pre},k}_{\ccalT_{m_j}}$ is the prefix part of the motion plan $\tau_{\ccalT_{m_j}}^k\models\phi_{\ccalT_{m_j}}$ constructed as per Section \ref{sec:autmodel}. The index $k$ is introduced in $\tau^{\text{pre},k}_{\ccalT_{m_j}}$ to point out that the prefix part is recomputed as the index $k$, introduced in (4) changes. The reason it needs to be recomputed is because the initial state of $\text{wPTS}_{\ccalT_{m_j}}$ changes as $k$ changes. Particularly, the initial state of $\text{wPTS}_{\ccalT_{m_j}}$ is the state $q^{0,k}_{\ccalT_{m_j}}=p_{\ccalT_{m_j}}^k(1)$. The state $\Pi|_{\text{wTS}_{i}}p_{\ccalT_{m_j}}^k(1)$ is selected so that all sub-paths $\Pi|_{\text{wTS}_{i}}p_{\ccalT_{m_j}}^k$ in $p_i^k$ chain up consistently. Hence, we select the state $\Pi|_{\text{wTS}_{i}}p_{\ccalT_{m_j}}^k(1)$ to be the final state of the previous sub-path that appeared in $p_i^k$. Consequently, the state $p_i^k(1)$ coincides with the the final state of the finite path $p_i^{k-1}$. If $k=1$, then $p_i^k(1)$ refers to the initial position of robot $i$ in the workspace.

The finite paths $p_{i}^k$ are constructed sequentially across the nodes $\ell_j\in\ccalV$, as follows. Let $\ccalS=\{\ell_1,\dots,\ell_j,\dots\}$ be an ordered sequence of the nodes in the mobility graph $\ccalG$, so that consecutive nodes (communication points) $\ell_j$, $\ell_e$ in $\ccalS$ are associated with teams $\ccalT_n$ and $\ccalT_m$, respectively, that belong to neighboring nodes in the graph $\ccalG_{\ccalT}$, i.e.,  $\ell_e\in\ccalC_m$, $\ell_j\in\ccalC_n$ and $m\neq n$ and $m\in\ccalN_{\ccalT_n}$. We assume that $\ccalS$ is known by all robots and that every robot $i$ is initially located at the first communication point 
$\ell_e\in\ccalC_m$, $m\in\ccalS_i$ that appears in $\ccalS$. Assume that paths have been constructed for all nodes in $\ccalS$ that precede $\ell_e\in\ccalC_m$ and that currently all robots $i\in\ccalT_m$ are located at node $\ell_e$ and coordinate to construct the paths $p_{i}^k$. Since the mobility graph $\ccalG$ is connected consecutive nodes in $\ccalS$ are connected by a path in $\ccalG$, this means that there is at least one robot $j\in\ccalT_n\cap\ccalT_m$, $n\in\ccalN_{\ccalT_m}$, which previously constructed its path $p_{j}^k$ by placing at its $n_{j}^{\ccalT_m}$-th entry of the finite path $\Pi|_{\text{wTS}_{j}}p_{\ccalT_m}^k$, i.e., $p_{j}^{k,n_{j}^{\ccalT_m}}=\Pi|_{\text{wTS}_{j}}p_{\ccalT_m}^k$. Then robot $i$ constructs the path $p_{i}^k$ based on three rules. According to the first rule, the path $\Pi|_{\text{wTS}_{i}}p_{\ccalT_m}^k$ will be placed at the $n_{i}^{\ccalT_m}$-th entry, which is selected to be equal to $n_{j}^{\ccalT_m}$, which is common for all robots $j\in\ccalT_m$ [line 1, Alg. \ref{alg:plan}]. This ensures that robot $i$ and all other robots $j\in\ccalT_m$ will meet at a communication point that belongs to $\ccalC_m$, as it will be shown in Proposition \ref{prop:satsynch}. The next step is to place the paths $\Pi|_{\text{wTS}_{i}}p_{\ccalT_g}^k$ for all $g\in\ccalS_i\setminus\{m\}$, at the $n_{i}^{\ccalT_g}$-th entry of $p_{i}^{k}$. The index $n_{i}^{\ccalT_g}$ will be determined by one of the two following rules. Specifically, according to the second rule if there exist robots $j\in\ccalN_{i}\cap\ccalT_g$ that have already constructed the paths $p_{j}^k$, then the index $n_{i}^{\ccalT_g}$ is selected to be equal to $n_{j}^{\ccalT_g}$, which is common for all $j\in\ccalT_g$ [line 4, Alg. \ref{alg:plan}]. Otherwise, according to the third rule the path $\Pi|_{\text{wTS}_{i}}p_{\ccalT_g}^k$ can be placed  at any free entry of $p_{i}^k$ indexed by $n_{i}^{\ccalT_g}$, provided that the $n_{i}^{\ccalT_g}$-th entry of all paths $p_{b}^k$ of robots $b\in\ccalN_{i}$ that have already been constructed does not contain states $\Pi|_{_{\text{wTS}_{b}}}p_{\ccalT_h}(k)$ with $h\in\ccalN_{\ccalT_g}$ [line 6, Alg. \ref{alg:plan}]. To highlight the role of this rule assume that $h\in\ccalN_{\ccalT_g}$. Then this means that there exists at least one robot $r\in\ccalT_h\cap\ccalT_g$. Then notice that without the third rule [line 5], at a subsequent iteration of this procedure, robot $r\in\ccalT_h\cap\ccalT_g$ would have to place the paths $\Pi|_{\text{wTS}_{r}}p_{\ccalT_g}^k$ and $\Pi|_{\text{wTS}_{r}}p_{\ccalT_h}^k$ at a common entry of $p_{r}^k$, i.e., $n_{r}^{\ccalT_g}=n_{r}^{\ccalT_h}$, due to the two previous rules and, therefore, a conflicting behavior for robot $r$ would occur. In all the remaining entries of $p_{i}^k$, $X$s are placed [line 8, Alg. \ref{alg:plan}].\footnote{If $\ell_j=\ell_1\in\ccalC_q$, for some $q\in\ccalV_{\ccalT}$, then initially, a randomly selected robot $j\in\ccalT_q$ creates arbitrarily its path $p_{j}^k$ by placing the paths $\Pi|_{\text{wTS}_{j}}p_{\ccalT_q}^k$  at the $n_{j}^{\ccalT_q}$-th entry of $p_{j}^k$, for all $q\in\ccalS_j$. Then the procedure previously described follows. Moreover, depending on the structure of the graph $\ccalG_{\ccalT}$ it is possible that a communication point $\ell_j\in\ccalC_m$ appears more than once in $\ccalS$. In this case, robots $i\in\ccalT_m$ construct the finite paths $p_{i}^k$ only the first time that $\ell_j$ appears in $\ccalS$.} This procedure is repeated until all robots $i\in\ccalT_m$ have constructed their respective paths $p_{i}^k$. Once this happens, all robots $i\in\ccalT_m$ depart from node $\ell_e\in\ccalC_m$ and travel to the next communication point $\ell_c\in\ccalC_v$, that appears in $\ccalS$ that satisfies $v\in\ccalS_i$  [line 9, Alg. \ref{alg:plan}]. At that point, all robots associated with the next communication point in $\ccalS$ are present at that node, and can coordinate to compute their respective paths, as before. The procedure is repeated sequentially over the nodes in $\ccalS$ until all robots have computed their paths. 

When all robots have constructed their finite paths, they exchange a set of indices denoted by $\ccalX_{i}$ that collects the indices $n_{i}^X$ at which $p_{i}^{k,n_{i}^X}=X$. If there exist paths $p_{i}^{k,n_{i}^X}=X$, for some $n_{i}^X\in\bigcap_{\forall i}\ccalX_{i}$, they are discarded, since in these paths all robots $i$ wait at their current states. Also, notice that in general, the finite paths $p_i^{k,e}$ for some $e\in\{1,\dots,\ell\}$ may have different lengths across the robots $i$. Consequently, this implies that two robots $i,j$ that belong to a team $\ccalT_m$ may start executing the finite paths $\Pi|_{\text{wTS}_{i}}p_{\ccalT_m}^k$ and $\Pi|_{\text{wTS}_{j}}p_{\ccalT_m}^k$ at different time instants, assuming that the robots pick synchronously their next states in their transition systems. Avoiding such a case is crucial to ensure intermittent communication within team $\ccalT_m$, as it will be shown in Proposition \ref{prop:satsynch}. Therefore, given any index $e$ we can introduce states at the end of the finite paths $p_i^{k,e}$ where the robots wait in their current states so that all finite paths $p_i^{k,e}$, for all robots $i$, have the same length. Communication between the robots in the last two stages of the algorithm can happen in the order defined by $\ccalS$, as before. 

\begin{rem}
Note that communication according to $\ccalS$ is very predictable and inefficient as it, e.g., does not allow for simultaneous meetings at the nodes of $\ccalG$. For these reasons it is only used to construct conflict-free motion plans that allow for much more efficient intermittent communication between robots. 
\end{rem}

\begin{rem}[Optimality]
To construct the motion plans $\tau_i$ in \eqref{eq:planij}, we first decouple the global LTL formula \eqref{eq:globalLTL} into local LTL expressions $\phi_{\ccalT_m}$. Then, we construct finite paths $p_{\ccalT_{m_j}}^k=\tau_{\ccalT_{m}}^{\text{pre},k}$, where $\tau_{\ccalT_{m}}^k\models\phi_{\ccalT_m}$, for all $m_j\in\ccalS_i$, and their concatenation (Algorithm \ref{alg:plan}) gives rise to the motion plans $\tau_i$. In principle, the paths $p_{\ccalT_{m_j}}^k$ connect the current configuration of robots $i\in\ccalT_{m_j}$ to the closest meeting point $\ell_e\in\ccalC_{m_j}$ where all robots in $\ccalT_{m_j}$ meet. Therefore, while these paths minimize the total distance traveled between the current and next meeting points, they do not optimize the infinite horizon cost function \eqref{eq:cost1}. Consequently, the proposed solution is suboptimal. To obtain an optimal solution that minimizes \eqref{eq:cost1}, standard model checking techniques can be applied to the global product system, that are known to be computationally expensive. Suboptimality here is a consequence of problem decomposition.

\end{rem}

\begin{algorithm}[t]
\caption{Construction of motion plans $\tau_{{i}}=[p_{i}^k]_{k=1}^{\infty}$ at node $\ell_e\in\ccalC_m$}
\label{alg:plan}
\begin{algorithmic}[1]
\REQUIRE Already constructed finite paths $p_{j}^k$ of robots $j\in\ccalN_{i}$;
\REQUIRE All robots in $\ccalT_m$ are located at node $\ell_e\in\ccalC_m$; 
\STATE $p_{i}^{k,n_{i}^{\ccalT_m}}:=\Pi|_{\text{wTS}_{i}}p_{\ccalT_m}^k$, $n_{i}^{\ccalT_m}=n_{j}^{\ccalT_m},~\forall j\in\ccalT_m$;
\FOR {$g\in\ccalS_i\setminus\{m\}$}
\IF {there exist constructed paths $p_{j}^k$, $j\in\ccalN_i\cap\ccalT_g$}
\STATE $p_{i}^{k,n_{i}^{\ccalT_g}}:=\Pi|_{\text{wTS}_{i}}p_{\ccalT_g}^k$, $n_{i}^{\ccalT_g}=n_{j}^{\ccalT_g},~\forall j\in\ccalT_g$;
\ELSE
\STATE $p_{i}^{k,n_{i}^{\ccalT_g}}:=\Pi|_{\text{wTS}_{i}}p_{\ccalT_g}^k$ provided either $p_{j}^{k,n_{i}^{\ccalT_g}}=X$, or $p_{j}^{k,n_{i}^{\ccalT_g}}=\Pi|_{\text{wTS}_{j}}p_{\ccalT_h}^k$ with $h\notin\ccalN_{\ccalT_g}$, $~\forall j\in\ccalN_{i}$;
\ENDIF
\STATE Put $X$s in the remaining entries;
\STATE Transmit path $p_{i}^k$ to a robot in $\ccalT_m$ that has not constructed its motion plan. If there are not such robots, all robots $i\in\ccalT_m$ depart from node $\ell_e\in\ccalC_m$;
\ENDFOR
\end{algorithmic}
\end{algorithm}

\subsection{Correctness of the Proposed Algorithm}
In this section, we show that the composition of the discrete motion plans $\tau_{i}$ generated by Algorithm \ref{alg:plan} satisfies the global LTL expression \eqref{eq:globalLTL}, i.e., that the network is connected over time. To prove this result, we need first to show that Algorithm \ref{alg:plan} can develop non-conflicting motion plans $\tau_{i}$, for which we have the following two results. Proposition \ref{prop:ell} can be proved by following the steps of the proof of Proposition 3.2 in \cite{kantaros16acc} for the graph $\ccalG_{\ccalT}$. Proposition \ref{prop:confl} holds by construction of the finite paths $p_i^k$ and its proof, which is omitted due to space limitations, is along the lines of the proof of Proposition 3.3 in \cite{kantaros16acc}.

\begin{prop}\label{prop:ell}
Algorithm \ref{alg:plan} can always construct finite paths $p_{i}^k$ that consist of $\ell$ finite paths $p_i^{k,e}$ where  $\ell\leq\text{max}\left\{d_{\ccalT_m}\right\}_{m=1}^M+1$.
\end{prop}

Proposition \ref{prop:ell} shows also that the finite paths $p_{i}^k$ and consequently, the motion plans $\tau_{i}$ depend on the node degree of graph $\ccalG_{\ccalT}$, and not on the size of the network. 

\begin{prop}\label{prop:confl}
Algorithm \ref{alg:plan} generates \textit{admissible} discrete motion plans $\tau_{i}$, i.e., motion plans that are free of conflicts and satisfy the transition rule $\rightarrow_{i}$ defined in Definition \ref{defn:wTS}.
\end{prop}

\begin{prop}\label{prop:satsynch}
The composition of motion plans $\tau_{i}$ generated by Algorithm \ref{alg:plan} satisfies the global LTL expression \eqref{eq:globalLTL}, i.e., connectivity of the robot network is ensured over time, infinitely often.
\end{prop}

\begin{proof}
To prove this result, it suffices to show that $\tau_{\ccalT_m,*}\models\phi_{\ccalT_m}$, for all teams $m\in\ccalV_{\ccalT}$, where $\phi_{\ccalT_m}$ is defined in \eqref{eq:localLTL} and $\tau_{\ccalT_m,*}=\otimes_{\forall i\in\ccalT_m}\tau_{i}$, where the motion plans $\tau_i$ are generated by Algorithm \ref{alg:plan}. Equivalently, according to Definition \ref{def:motionplan} it suffices to show that
\begin{equation}\label{eq:trace0}
\texttt{trace}(\tau_{\ccalT_m,*})\in\texttt{Words}(\phi_{\ccalT_m}).
\end{equation}
First, assume that the prefix part $\tau_{\ccalT_m}^{\text{pre},k}$ of the motion plan $\tau_{\ccalT_m}^k=\tau_{\ccalT_m}^{\text{pre},k}[\tau_{\ccalT_m}^{\text{suf},k}]^{\omega}\models\phi_{\ccalT_m}$ has been constructed, for all $k$, using a standard model checking method as described in Section \ref{sec:autmodel}.\footnote{Note that, throughout this proof, the only difference between the plans $\tau_{\ccalT_m,*}$ and $\tau_{\ccalT_m}$ is that the first one has been derived through composing all non-conflicting motion plans $\tau_{i}$, $\forall i\in\ccalT_m$ generated by Algorithm \ref{alg:plan}, while the second one is the plan for all robots $i\in\ccalT_m$ computed as described in Section \ref{sec:autmodel}.} By construction of the motion plans $\tau_i$, it holds that there is an index $n_i^{\ccalT_m}$, $1\leq n_i^{\ccalT_m}\leq \ell$, such that $p_i^{k,n_i^{\ccalT_m}}=\Pi|_{\text{wTS}_i}p_{\ccalT_m}^k$ and an index $e_i^{\ccalT_m}$, which are common for all robots $i\in\ccalT_m$, such that $p_i^k(e_i^{\ccalT_m})=\Pi|_{\text{wTS}_i}p_{\ccalT_m}^k(1)$ for all $i\in\ccalT_m$. Consequently, in the finite path $\pi_{\ccalT_m,*}^{k}=\otimes_{\forall i\in\ccalT_m}p_{i}^k$ there are indices $n^{\ccalT_m}$ and $e^{\ccalT_m}$ such that $\pi_{\ccalT_m,*}^{k,n^{\ccalT_m}}=p_{\ccalT_m}^k$ and $\pi_{\ccalT_m,*}^{k}(e^{\ccalT_m})=p_{\ccalT_m}^k(1)$ for all $k$. Hence, we conclude that all robots $i\in\ccalT_m$ will start executing the finite path $p_{\ccalT_m}^k$, simultaneously.  

Using the above observation, we examine the properties of the finite word that is generated by the team $\ccalT_m$ when the finite path $p_{\ccalT_m}^k$ is executed. Let $w_{\ccalT_m}^k\in(2^\mathcal{AP})^*$ be a finite word defined as $w_{\ccalT_m}^k=\texttt{trace}(\tau_{\ccalT_m}^{\text{pre},k})$ or, equivalently, by construction of the finite paths $p_{\ccalT_m}^k$, $w_{\ccalT_m}^k=\texttt{trace}(p_{\ccalT_m}^k)$. Since $\tau_{\ccalT_m}^k\models\phi_{\ccalT_m}$ and by construction of the prefix part $\tau_{\ccalT_m}^{\text{pre},k}$ for all $k$, we have that the infinite word $\sigma_{\ccalT_m}=w_{\ccalT_m}^0w_{\ccalT_m}^1w_{\ccalT_m}^2\dots\in(2^\mathcal{AP})^{\omega}$ satisfies
\begin{equation}\label{eq:sigmaTm}
\sigma_{\ccalT_m}\in\texttt{Words}(\phi_{\ccalT_m}). 
\end{equation}
Next, notice that between the execution of the finite paths $p_{\ccalT_m}^k$ and $p_{\ccalT_m}^{k+1}$ robots $i\in\ccalT_m$ will traverse through some states of their respective wTSs that are determined by $\tau_{\ccalT_m,*}$ until the state $p_{\ccalT_m}^{k+1}(1)$ is reached. In other words, the team of robots $\ccalT_m$ does not execute consecutively the finite paths $p_{\ccalT_m}^k$ and $p_{\ccalT_m}^{k+1}$. Now, we need to show that these intermediate transitions that robots make are admissible in $\text{wPTS}_{\ccalT_m}$ and do not violate $\phi_{\ccalT_m}$. These transitions are admissible as implied by Proposition \ref{prop:confl}. Also, they cannot violate $\phi_{\ccalT_m}$, since the LTL expressions $\phi_{\ccalT_m}$ in \eqref{eq:localLTL}, for all $m\in\ccalV_{\ccalT}$, do not include the negation operator $\neg$ in front of the atomic propositions $\pi_{{i}}^{\ell_j}$ defined in Section \ref{sec:intermittent}. Therefore, this equivalently means that as the robots $i\in\ccalT_m$ execute the motion plan $\tau_{\ccalT_m,*}$ they will eventually pass through all states determined by the finite paths $p_{\ccalT_m}^k$, for all $k$, without violating $\phi_{\ccalT_m}$. Consequently, this means that as the robots $i\in\ccalT_m$ move according to the motion plan $\tau_{\ccalT_m,*}$, the generated trace $\texttt{trace}(\tau_{\ccalT_m,*})$ will certainly include the atomic propositions that are included in $\sigma_{\ccalT_m}$; see also Definition \ref{def:tracet}. In other words, the sequence of atomic propositions $\sigma_{\ccalT_m}$ is a subsequence of $\texttt{trace}(\tau_{\ccalT_m,*})$ while the additional atomic propositions that exist in $\texttt{trace}(\tau_{\ccalT_m,*})$ cannot violate $\phi_{\ccalT_m}$, as previously discussed. Therefore, due to \eqref{eq:sigmaTm}, we conclude that \eqref{eq:trace0} holds for all $m\in\ccalV_{\ccalT}$, which completes the proof.\end{proof}

In general, the motion plans $\tau_{i}=[\tau_i(n)]_{n=1}^{\infty}$ defined in \eqref{eq:planij} are infinite paths of $\text{wTS}_{i}$, since the finite paths $p_i^k$ need to be updated for every $k\in\mathbb{N}$. Therefore, in practice they are hard to implement and manipulate. In the following proposition, we show that the motion plans $\tau_{i}$ constructed by Algorithm \ref{alg:plan} have a finite representation and they can be expressed in a prefix-suffix structure, where the prefix part $\tau_{i}^{\text{pre}}$ is traversed only once and the suffix part $\tau_{i}^{\text{suf}}$ is repeated infinitely.

\begin{prop}\label{prop:presuf}
Algorithm \ref{alg:plan} generates discrete motion plans $\tau_{i}$ for all robots $i$ in a prefix-suffix structure, i.e., $\tau_{i}=\tau_{i}^{\text{pre}}\left[\tau_{i}^{\text{suf}}\right]^{\omega}=[p_i^1\dots p_i^{k_p-1}][p_i^{k_p}\dots p_i^{k_s}]^{\omega}$.
\end{prop}

\begin{proof}
To show this result it suffices to show that there is an index $k_s$, such that for all $k\geq k_s$, the finite paths $p_i^k$ are repeated and, therefore, they do not need to be recomputed. First, notice that such an index $k_s$ is common for all motion plans $\tau_i$. To illustrate this point assume that there is a robot $i\in\ccalT_m$ such that $\Pi|_{\text{wTS}_i}p_{\ccalT_m}^k\neq \Pi|_{\text{wTS}_i}p_{\ccalT_m}^{k+1}$ and, therefore, $p_i^k\neq p_i^{k+1}$. Since the paths $p_{\ccalT_m}^k$ are computed collectively by all robots $i\in\ccalT_m$, the fact that $\Pi|_{\text{wTS}_i}p_{\ccalT_m}^k\neq \Pi|_{\text{wTS}_i}p_{\ccalT_m}^{k+1}$ means that there may be another robot $j\in\ccalT_m$ such that $\Pi|_{\text{wTS}_j}p_{\ccalT_m}^k\neq \Pi|_{\text{wTS}_j}p_{\ccalT_m}^{k+1}$ and, consequently, $p_j^k\neq p_j^{k+1}$, which may be propagated to all other robots in the network, since graph $\ccalG_{\ccalT}$ is connected. Therefore, the index $k_s$ is common for all robots.

In what follows, we show the existence of an index $k_s$ defined above. By definition the finite paths $p_i^k$ are a concatenation of $\ell$ finite paths, where the $e$-th finite path in $p_i^k$ is denoted by $p_i^{k,e}$. Also, by construction we have that $p_i^{k,e}\neq X$ is constructed by computing a shortest path in a graph associated with a PBA $P_{\ccalT_m}$ from an initial to a final state in $P_{\ccalT_m}$. Notice that the set of final states $\ccalF_{P_{\ccalT_m}}$ in $P_{\ccalT_m}$ remain the same for all $k$, which is not the case for the set of initial states $\ccalQ^0_{P_{\ccalT_m}}=q^{0,k}_{\ccalT_m}\times\ccalQ_{B_{\ccalT_m}}^0$. The reason is that the initial state $q^{0,k}_{\ccalT_m}$ of $\text{wPTS}_{\ccalT_m}$ may change over $k$ and, specifically, we have that $q^{0,k}_{\ccalT_m}=p_{\ccalT_m}^k(1)=(p_{i_1}^{k,e}(1),\dots,p_{i_{|\ccalT_m|}}^{k,e}(1))$, where the construction of state $p_{\ccalT_m}^k(1)$ was presented in Section \ref{sec:autmodel}. Since the number of possible combinations of robots' states in their respective wTS is finite, we have that there is a finite number of possible initial states $q^{0,k}_{\ccalT_m}$ for all $m\in\ccalV_{\ccalT}$ and, consequently, a finite number of possible initial states $q^{0,k}=(p_{1}^{k,e}(1),\dots,p_{N}^{k,e}(1))$. Therefore, for any $e\in\{1,\dots,\ell\}$ there are always two indices $k_1^e$ and $k_2^e$, $k_2^e>k_1^e$ such that $p_i^{k_1^e,e}(1)=p_i^{k_2^e,e}(1)$, for all robots $i$, since otherwise that would mean that there are infinite number of possible $q^{0,k}=(p_{1}^{k,e}(1),\dots,p_{N}^{k,e}(1))$ for any $e\in\{1,\dots,\ell\}$. Since $p_i^{k_1^e,e}(1)=p_i^{k_2^e,e}(1)$ and the distances between communication points remains the same for all $k$, we have that $p_i^{k_1^e,e}=p_i^{k_2^e,e}$. Then, the finite paths $p_i^{k,e}$, for all $e\in\{1,2,\dots,\ell\}$, for all robots $i$, and for all $k\geq k_2^e$ have already been constructed at $k$ that satisfies $k_1^e\leq k \leq k_2^e$; for instance, for $k=k_2^e+1$ it holds that $p_i^{k_2^e+1,e}=p_i^{k_1^e+1,e}$. Let $e^*=\argmin_e\{k_2^e\}_{e=1}^{\ell}=\argmin_e\{k_1^e\}_{e=1}^{\ell}$. Then, similarly, we have that the finite paths $p_i^{k,e}$, for all $e\in\{1,2,\dots,\ell\}$, for all robots $i$, and for all $k\geq k_2^{e*}$ have already been constructed at $k$ that satisfies $k_1^{e^*}\leq k \leq k_2^{e^*}$. Consequently, this means that the finite paths $p_i^k$ for all $k\geq k_2^{e^*}$ have already been constructed as well at $k$ that satisfies $k_1^{e^*}\leq k \leq k_2^{e^*}$; for instance $p_i^{k_1^{e^*}}=p_i^{k_2^{e^*}}$ and $p_i^{k_1^{e^*}+1}=p_i^{k_2^{e^*}+1}$. Hence, the motion plans $\tau_i$ can be written in a prefix-suffix structure $\tau_{i}=\tau_{i}^{\text{pre}}\left[\tau_{i}^{\text{suf}}\right]^{\omega}$ where $\tau_{i}^{\text{pre}}=[p_i^1\dots p_i^{k_1^{e^*}-1}]$ and $\tau_{i}^{\text{suf}}=[p_i^{k_1^{e^*}}\dots p_i^{k_2^{e^*}-1}p_i^{k_2^{e^*}}]$, i.e., $k_p=k_1^{e^*}$ and $k_s=k_2^{e^*}$, which completes the proof.
\end{proof}


\section{Asynchronous Intermittent Communication}\label{sec:asynch}

In section \ref{sec:control}, we showed that if all robots $i$ pick synchronously their next states in $\text{wTS}_{i}$ according to the motion plans $\tau_{i}$, then the LTL expression \eqref{eq:globalLTL} is satisfied. In this section, we show that the generated motion plans can be executed asynchronously, as well, by appropriately introducing delays in the continuous-time execution of $\tau_{i}$. We omit a formal proof of this result due to space limitations, and instead validate the proposed asynchronous scheme through numerical simulations in Section \ref{sec:sim}.

Due to the asynchronous execution of the controllers, the motion plans $\tau_{i}$ can be written as in \eqref{eq:planij} replacing the indices $n$ and $k$ with $n_i$ and $k_i$, respectively, which allows us to model the situation where the robots pick asynchronously their next states in $\text{wTS}_{i}$.   
In the asynchronous execution of the infinite paths $\tau_{i}$ robot $i$ moves from state $\tau_{i}(n_{i}-1)\in\ccalQ_i$ to $\tau_{i}(n_{i})\in\ccalQ_i$ according to a continuous-time motion controller $\bbu_{i}(t)\in\mathbb{R}^n$ that belongs to the tangent space of $\gamma_{ij}$ at $\bbx_{i}(t)$. Without loss of generality, assume that $i\in\ccalT_m$ and $\tau_{i}(n_{i})=q_i^{\ell_e}$, for some $\ell_e\in\ccalC_m$. When robot $i$ arrives at state $\tau_{i}(n_{i})$ it checks if $\tau_{j}(n_i)=\Pi|_{\text{wTS}_j}p_{\ccalT_m}(f_i)=q_j^{\ell_e}$ for all $j\in\ccalT_m$. If this is the case, then robot $i$ waits at node $\ell_e$ until all other robots $j\in\ccalT_m$ are present there. When this happens, or if there is at least one robot $j\in\ccalT_m$ such that $\tau_{j}(n_i)\neq\Pi|_{\text{wTS}_j}p_{\ccalT_m}(f_i)=q_j^{\ell_e}$, then robot $i$ moves towards the next state $\tau_{i}(n_{i}+1)$.


\section{Simulation Studies}\label{sec:sim}

             
In this section, a simulation study is provided that illustrates our approach for a network of $N=5$ robots that move along the edges of the mobility graph with $L=20$ communication points as shown in Figure \ref{sim1}. The network is divided in $M=5$ teams which are $\ccalT_1=\{1,2\}$, $\ccalT_2=\{2,3\}$, $\ccalT_3=\{3,4\}$, $\ccalT_4=\{2,4,5\}$, and $\ccalT_5=\{1,5\}$ and, therefore, the graph $\ccalG_{\ccalT}$ is as shown in Figure \ref{graphgt}. Also, the mobility graph is constructed so that there is a path $\gamma_{ij}$ from any node $\ell_i$ to any other node $\ell_j$. Therefore, the finite paths $p_{\ccalT_m}^k$ constructed as per Section \ref{sec:autmodel} have the form $p_{\ccalT_m}^k=q^{0,k}_{\ccalT_m}q_{\ccalT_m}$, where $q^{0,k}_{\ccalT_m}$ is the initial state of $\text{wPTS}_{\ccalT_m}$ constructed as defined in Section \ref{sec:conflict} and $q_{\ccalT_m}=(q_{i_1}^{\ell_{j}}\dots q_{i_{|\ccalT_m|}}^{\ell_{j}})$ is a state where all robots of team $\ccalT_m$ are located at a common meeting point $\ell_j\in\ccalC_m$. The motion plans $\tau_i$ generated by Algorithm \ref{alg:plan} have the following structure: 

{\small\begin{align}\label{plans_robots}
\tau_{1}=&\left[p_{1}^{k_1}\right]_{k_1=1}^{\infty}=\left[\Pi|_{\text{wTS}_{1}}\tau_{\ccalT_1}^{k_1}\Pi|_{\text{wTS}_{1}}\tau_{\ccalT_5}^{k_1}X\right]_{k_1=1}^{\infty},\nonumber\\
\tau_{2}=&\left[p_{2}^{k_2}\right]_{k_2=1}^{\infty}=\left[\Pi|_{\text{wTS}_{2}}\tau_{\ccalT_1}^{k_2}\Pi|_{\text{wTS}_{2}}\tau_{\ccalT_2}^{k_2}\Pi|_{\text{wTS}_{2}}\tau_{\ccalT_4}^{k_2}\right]_{k_2=1}^{\infty},\nonumber\\
\tau_{3}=&\left[p_{3}^{k_3}\right]_{k_3=1}^{\infty}=\left[\Pi|_{\text{wTS}_{3}}\tau_{\ccalT_3}^{k_3}\Pi|_{\text{wTS}_{3}}\tau_{\ccalT_2}^{k_3}X\right]_{k_3=1}^{\infty},\nonumber\\
\tau_{4}=&\left[p_{4}^{k_4}\right]_{k_4=1}^{\infty}=\left[\Pi|_{\text{wTS}_{4}}\tau_{\ccalT_3}^{k_4}X\Pi|_{\text{wTS}_{4}}\tau_{\ccalT_4}^{k_4}\right]_{k_4=1}^{\infty},\nonumber\\
\tau_{5}=&\left[p_{5}^{k_5}\right]_{k_5=1}^{\infty}=\left[X\Pi|_{\text{wTS}_{5}}\tau_{\ccalT_5}^{k_5}\Pi|_{\text{wTS}_{5}}\tau_{\ccalT_4}^{k_5}\right]_{k_5=1}^{\infty},\nonumber
\end{align}}

which can be written in a prefix-suffix structure with $k_s=2$ and $k_p=1$, where the indices $k_s$ and $k_p$ are defined in Proposition \ref{prop:presuf}. The motion plans $\tau_i$ defined above are depicted in Figure \ref{sim1}. Notice that the distances between any two meeting points vary across $\ccalG$ and so do the robots' velocities and, therefore, robots pick asynchronously their next states in $\text{wTS}_i$. Consequently, this results in waiting times for every robot $i\in\ccalT_m$ at the meeting points $\ell_j\in\ccalC_m$, which are non-integer multiples of each other, for all $m\in\ccalV_{\ccalT}$. This illustrated in Figure \ref{fig:wait}, where, e.g., in team $\ccalT_4$, robots $2$, $4$, and $5$ wait at a meeting point $\ell_j\in\ccalC_4$ for $0$, $3.4$, and $0.8$ time units, respectively. Observe also in Figure \ref{fig:wait} that robot $2$ never waits at any meeting point.

To illustrate that under the proposed motion plans connectivity is ensured over time, we implement a consensus algorithm over the dynamic network $\ccalG_c$. Specifically, we assume that initially robots generate a random number $v_{i}(t_0)$ and when all robots $i\in\ccalT_m$ meet at $\ell_j\in\ccalC_m$ they perform the following consensus update: $v_{i}(t)=\frac{1}{\left|\ccalT_m\right|}\sum_{e\in\ccalT_m}v_{e}(t)$. Figure \ref{fig:cons} shows that eventually all robots reach a consensus on the numbers $v_{i}(t)$. Note also that applying existing LTL-based planning would result in PBA constructed for all robots with $\Pi_{i=1}^N|\ccalQ_i||\ccalQ_{B}|=45360|\ccalQ_B|$ states, where $\ccalQ_B$ corresponds to the state-space of the NBA associated with the global LTL expression in \eqref{eq:globalLTL}, which is hard to manipulate in practice. This issue becomes more severe as the size of the network increases.
\begin{figure}[t]
  \centering
     \includegraphics[width=0.7\linewidth]{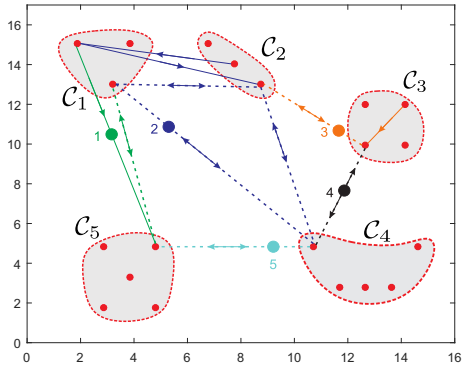}
  \caption{Intermittent communication of $N=5$ robots moving along the edges of an underlying mobility graph. Red dots represent communication points in each set $\ccalC_m$ while edges between any two communication points exist (not shown). Straight lines depict the robots' trajectories as determined by motion plans $\tau_i$. The suffix part of $\tau_i$ is represented by dashed lines for every robot $i$ while solid lines depict a part of prefix structure, i.e., a path that connects robots' initial states to the respective suffix structure. The prefix part of robots 4 and 5 coincides with their respective suffix part and, therefore, there are no corresponding solid lines for them.}
  \label{sim1}
\end{figure}
\begin{figure}[t]
  \centering
     \subfigure[Waiting Time]{
    \label{fig:wait}
  \includegraphics[width=0.47\linewidth]{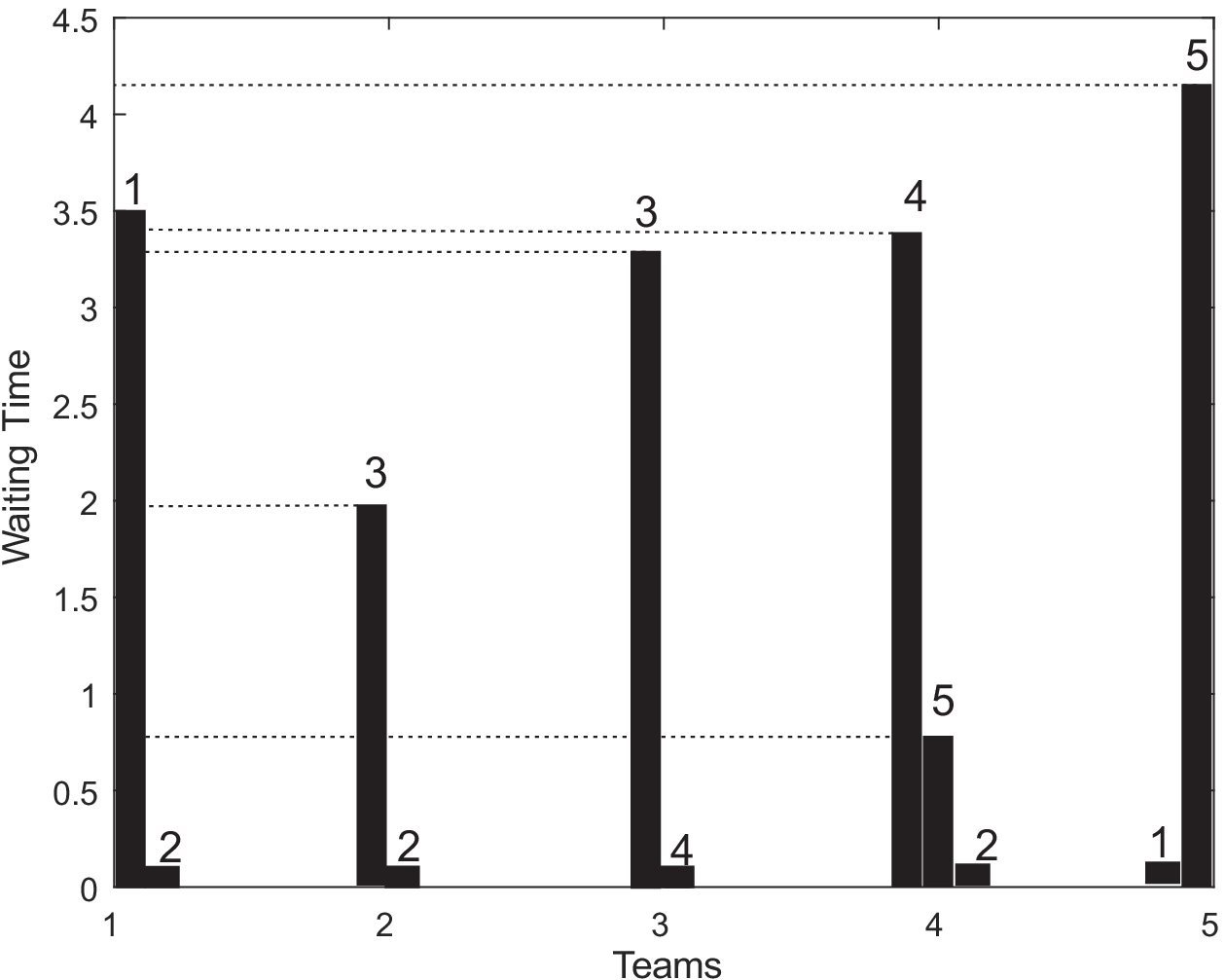}}
  \subfigure[Consensus]{
    \label{fig:cons}
  \includegraphics[width=0.46\linewidth]{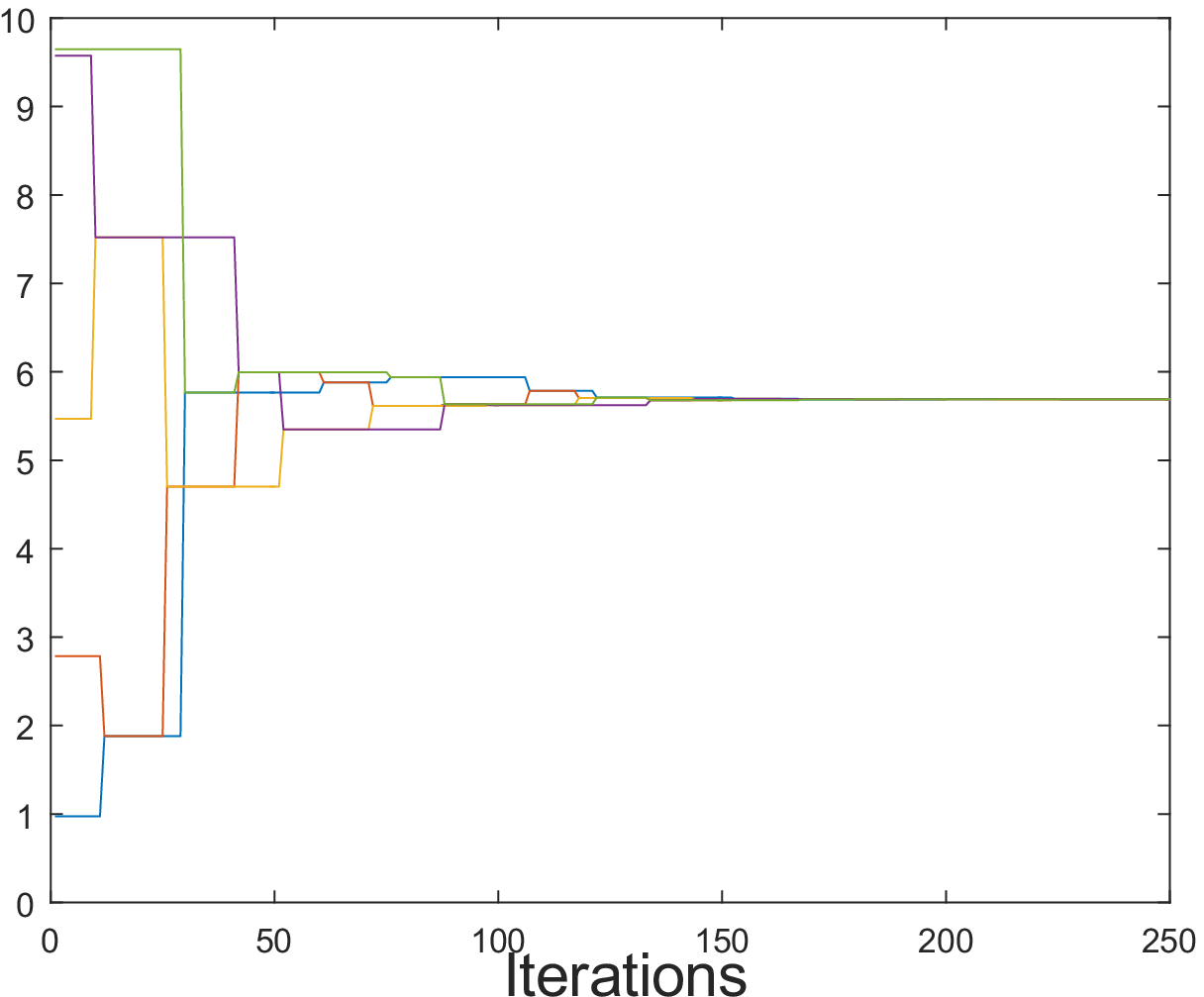}}
  \caption{Figure \ref{fig:wait} depicts the waiting time of every robot $i$ at a meeting point $\ell_j\in\ccalC_m$, for all $m\in\ccalS_i$ during a single execution of $\tau_i^{\text{suf}}$. Figure \ref{fig:cons} illustrates the consensus of numbers $v_{i}(t)$.}
  \label{fig:plots}
\end{figure}
\section{Conclusion}\label{sec:conclusion}
In this paper we considered the problem of controlling intermittent communication in mobile robot networks. We assumed that robots move along the edges of a mobility graph and they can communicate only when they meet at its nodes, which gave rise to a dynamic communication network. The network was defined to be connected over time if communication takes place at the rendezvous points infinitely often which was captured by an LTL formula. Then this LTL expression was approximately decomposed into local LTL expressions which were assigned to robots. To avoid conflicting robot behaviors that could occur due to this approximate decomposition, we developed a distributed conflict resolution scheme that generated free-of-conflicts discrete motion plans for every robot that ensured connectivity over time, infinitely often and minimized the distance traveled by the robots. We also showed that the generated motion plans can be executed asynchronously by introducing delays in their continuous-time execution. 
\bibliographystyle{IEEEtran}
\bibliography{YK_bib}
\end{document}